\documentclass[11pt]{article}

\usepackage{amsmath}
\usepackage{amssymb}
\usepackage{latexsym}
\usepackage{array}
\usepackage{stmaryrd}
\usepackage{paralist}
\usepackage{graphicx} 
\usepackage{caption}
\usepackage{amsthm}
\usepackage{xspace}
\usepackage{subfig}
\usepackage{url}
\usepackage{tabulary,booktabs,multirow}
\usepackage{algorithm, algorithmic}

\usepackage{tikz}
\usetikzlibrary{positioning,arrows,automata,shapes,decorations.pathreplacing}
\usepackage{footnote}
\usepackage{subfig}
\usepackage{xcolor}

\newcommand{\set}[1]{\ensuremath{\left\{#1\right\}}} 

\newcommand{\dsod}{\mathsf{atleast}}

\newcommand{\ie}{\emph{i.e.}\xspace}

\newcommand{\atmost}{\mathsf{atmost}}

\newcommand{\card}[1]{\left|#1\right|}

\newcommand{\sqconcat}{\mathbin{+}}
\newcommand{\sqinter}{\mathbin{\ast}}

\newcommand{\WSPlong}{\textsc{Workflow Satisfiability Problem}\xspace}
\newcommand{\WSPshort}{\textsc{WSP}\xspace}
\newcommand{\WSP}{\WSPshort}

\newcommand{\WSRPPlong}{\textsc{WSP with Release Points}\xspace}

\newcommand{\WSRPP}{\WSRPPlong}

\newcommand{\ccws}{CCWS\xspace}

\newcommand{\Bxor}{\mathcal{B}}

\newcommand{\cO}{\mathcal{O}}
\renewcommand{\l}{{\rm left}}
\renewcommand{\r}{{\rm right}}
\renewcommand{\between}{{\rm btw}}

\newcommand{\GGrem}[1]{{ #1}}
\newcommand{\GGadd}[1]{{ #1}}

 \newtheorem{definition}{Definition}
 \newtheorem{theorem}{Theorem}
 \newtheorem{lemma}{Lemma}

\begin{document}

\title{Bounded and Approximate Strong Satisfiability in Workflows}

\author{Jason Crampton, Gregory Gutin, Diptapriyo Majumdar\\
Royal Holloway, University of London, Egham, United Kingdom}

\maketitle%

\begin{abstract}
There has been a considerable amount of interest in recent years in the problem of workflow satisfiability, which asks whether the existence of constraints in a workflow specification makes it impossible to allocate authorized users to each step in the workflow.
Recent developments have seen the workflow satisfiability problem (WSP) studied in the context of workflow specifications in which the set of steps may vary from one instance of the workflow to another.
This, in turn, means that some constraints may only apply to certain workflow instances.
Inevitably, WSP becomes more complex for such workflow specifications.
Other approaches have considered the possibility of associating costs with the violation of ``soft'' constraints and authorizations.
Workflow satisfiability in this context becomes a question of minimizing the cost of allocating users to steps in the workflow.
In this paper, we introduce new problems, which we believe to be of practical relevance, that combine these approaches.
In particular, we consider the question of whether, given a workflow specification with costs and a ``budget'', all possible workflow instances have an allocation of users to steps that does not exceed the budget.
We design a fixed-parameter tractable algorithm to solve this problem parameterized by the total number of steps, release points and xor branchings.
\end{abstract}

%
%
%

\section{Introduction}\label{sec:intro}

Many businesses use computerized systems to manage their business processes.
A common example of such a system is a workflow management system, which is responsible for the co-ordination and execution of steps in a business process.
The overall structure of the business process is fixed and may be defined as a workflow specification comprising a set of steps that must be performed in a particular sequence.
However, the specific steps that are performed may vary from one instance of the workflow to another.
For example, certain steps may only be relevant in a purchase order workflow if the value of a specific order exceeds a particular value.

The steps in a workflow instance are executed by users and these users will vary from instance to instance.
For most workflows, we may wish to impose some form of access control on the execution of those steps, limiting which users may perform which steps.
This control may take the form of an authorization policy and a set of authorization constraints: the former defines which users are authorized to perform which steps; and the latter limits the combinations of users that may perform certain sets of steps in the business process.
A simple form of constraint prohibits the same user from performing two (or more) particular security-sensitive steps.
 
As we noted above, the steps executed in a workflow may vary from one instance to another.
Similarly, there may be constraints that only apply when certain sub-workflows are executed in a particular workflow instance.
Basin, Burri and Karjoth introduced a mechanism for modeling such constraints using \emph{release points}~\cite{BaBuKa14}.
Informally, release points allow a constraint to be ``switched off'' when some specified points in a workflow instance are reached. 
In particular, when different release points are located in different mutually exclusive sub-processes, it is possible to encode \emph{conditional} constraints.

An assignment of users to workflow steps is a \emph{plan}.
A plan is \emph{valid} if all users are authorized and no constraint is violated.
We say a workflow specification is \emph{satisfiable} if there exists a \emph{valid plan} for the specification~\cite{WaLi10}.
An efficient algorithm for deciding the so-called \emph{workflow satisfiability question} (WSP) is important from the point of view of static analysis of workflow specifications and as an on-line access control decision problem~\cite[Section 2.2]{CrGu13}.

Crampton, Gutin and Karapetyan~\cite{CrGuKa15b} introduced the \emph{valued} workflow satisfiability problem (VWSP).
Informally, constraint and authorization violations are associated with costs, which may be regarded as an estimate of the risk associated with allowing those violations.
A solution to VWSP is an assignment of users to steps having minimal cost, this cost being zero when the workflow is satisfiable.

In this paper, we introduce a new class of workflow satisfiability problems, based on extended workflow specifications and costs associated with policy and constraint violation.
The notion of \emph{strong} satisfiability is associated with workflow specifications in which the set of steps executed in a workflow instance may vary from instance to another.
Typically, this variation arises because of conditional branching in the specification.
Strong satisfiability requires that there exists a valid plan for every possible set of steps that could form a workflow instance.
A weaker form of strong satisfiability asks whether there exists a ``reasonable'' plan for every possible set of steps, where ``reasonable'' means its cost does not exceed some threshold value that is part of the input to the problem.
An alternative question would be that of ``approximate satisfiability'', which asks whether there exists a valid plan for at least some fraction (stated as part of the input to the problem) of all possible 
sets of  steps.
An organization might be prepared to make a workflow specification operational if, for example, it is known that there exists a valid plan for at least 99\% of the possible workflow instances.

We define three decision problems based on the informal notions described above and show that these problems are fixed-parameter tractable, subject to reasonable assumptions about the workflow specification.  Informally, this means that relatively efficient algorithms exist to solve these problems.

In Section~\ref{sec:notation}, we introduce relevant notation and terminology.
Then, in Section~\ref{sec:wsp}, we describe relevant workflow models and satisfiability problems.
We formally define bounded and approximate WSP problems in Section~\ref{sec:bounded-wsp}, then prove these problems are fixed-parameter tractable in Sections~\ref{sec:bounded-cost-wsp1} and~\ref{sec:bounded-expected-cost}.
The paper concludes with sections describing related work, a summary of our contributions and our ideas for future research in this area. In Appendix, we provide a running example illustrating some key concepts of this paper.

\section{Notation and terminology}\label{sec:notation}

\subsection{Directed graphs}
A directed graph ({\em digraph} for short) is a pair $G=(V, E)$, where $V$ is the set of vertices, and $E \subseteq V \times V$ is the set of edges. 
A directed acyclic graph (DAG) is a digraph which does not contain any directed cycle, \ie no sequence $(u_0, u_1 \dots, u_{k-1}, u_0)$ such that each pair of consecutive vertices belongs to $E$. 

For $u \in V$, we define the {\em in-neighborhood} of $u$ to be the set $N^-(u) = \{t \in V : (t, u) \in E\}$; the {\em in-degree} of $u$ is the size of its in-neighborhood $\card{N^-(u)}$.
Similarly, the {\em out-neighborhood} of $u$ is the set $N^+(u) = \{w \in V : (u, w) \in E\}$ and the \emph{out-degree} of $u$ is $\card{N^+(u)}$.
A vertex of in-degree $0$ is called a \emph{source}, while a vertex of out-degree $0$ is called a \emph{sink}.
For $S \subseteq V$, we denote by $G[S]$ the \emph{induced subgraph} $(S, E \cap (S \times S))$.
By abuse of notation, we will sometimes write $G \setminus S$ as a shortcut for $G[V \setminus S]$.
For additional information about directed and undirected graphs and DAGs, we refer the reader to~\cite{BJGu02,Die12,Gutin18}.

Sometimes, it is convenient to represent a DAG with a partial order on its vertices. 
Indeed, we may write $u \le v$ for $u, v \in V$ whenever $u=v$ or there exists a directed path from $u$ to $v$. 
By extension, we may write $u < v$ if $u \le v$ and $u \neq v$. 

For any positive integer $n$, let $[n] = \{1, \dots, n\}$.
An ordered sequence $\sigma = (v_1, \dots, v_q)$ of distinct vertices of $V$ is called a \emph{linear subextension} of $G$ if and only if for every $i, j \in [q]$, $v_i \le v_j$ implies $i \le j$. If $\sigma$ contains all vertices of $V$, then we say that $\sigma$ is a \emph{linear extension} of $G$.

\subsection{Fixed-parameter tractability}
Many decision problems take several parameters as input.
It can be instructive to consider how the complexity of the problem may change if we assume one or more of those parameters is small relative to the others.
The purpose of multivariate analysis of the complexity of a problem is to obtain efficient algorithms when the chosen parameters take small values in practice. 
We say that a decision problem is \emph{fixed-parameter tractable} (FPT) if there exists an algorithm that decides if an instance is positive in $\cO(f(\kappa)p(n))$ time for some computable function $f$ and some polynomial $p$, where $n$ denotes the size of an instance, and $\kappa$ is a parameter of the instance. 
Accordingly, we will call such an algorithm an \emph{FPT algorithm}. 

In many cases, the decision problem under consideration has several parameters $\kappa_1,\dots ,\kappa_p$. This is reduced to the case of just one parameter by considering the parameter $\kappa=\kappa_1+\dots + \kappa_p.$ In fixed-parameter tractable algorithmics, the time $\cO(f(\kappa)p(n))$ is often written as $\cO^*(f(\kappa)).$ Thus, $\cO^*$ hides not only constant factors like $\cO$, but also polynomials (exponential functions are usually more important when evaluating the running time of FPT algorithms than polynomial factors).
For more details about parameterized complexity, we refer the reader to the monographs of Downey and Fellows~\cite{DoFe13} and Cygan {\em et al.}~\cite{CyganFKLMPPS15}.

\section{The workflow satisfiability problem}\label{sec:wsp}
A \emph{workflow specification} is defined by a directed acyclic graph $G=(S, E)$, where $S$ is the set of steps to be executed, and $E \subseteq S \times S$ defines a partial ordering on the set of steps in the workflow, in the sense that $(s_1, s_2) \in E$ means that step $s_1$ must be executed before $s_2$ in every instance of the workflow. 
Note that the order is not required to be total, so the exact sequence of steps may vary from instance to instance.
In addition, we are also given a set of users $U$ and an \emph{authorization policy} $A \subseteq S \times U$, where $(s, u) \in A$ means that user $u$ is authorized to execute step $s$. 
A workflow specification $G=(S,E)$ together with an authorization policy is called a \emph{workflow schema}. 
Throughout the paper, we will assume that for every step $s \in S$, there exists some user $u \in U$ such that $(s, u) \in A$.

A workflow \emph{constraint} $(T, \Theta)$ limits the users that are allowed to perform a set of steps $T$ in any execution of the workflow. In particular, $\Theta$ identifies authorized (partial) assignments of users to steps in $T$,
 i.e. $\Theta$ is a set of functions from $T$ to $U$.
A (partial) plan is a function $\pi : S' \rightarrow U$, where $S' \subseteq S$.
A plan $\pi : S \rightarrow U$ represents an allocation of steps to users.
The workflow satisfiability problem (WSP) is concerned with the existence or otherwise of a plan that is authorized and satisfies all constraints.

More formally, let $\pi : S' \rightarrow U$, where $S' \subseteq S$, be a plan. Given $T \subseteq S'$, we write $\pi|_{T}$ to denote the function $\pi$ restricted to domain $T$; that is $\pi|_{T} : T \rightarrow U$ is defined by $\pi|_{T}(s) = \pi(s)$ for all $s \in T$. 
Then we say $\pi : S' \rightarrow U$ \emph{satisfies} a workflow constraint $(T, \Theta)$ if $T \not \subseteq S'$ or $\pi|_{T} \in \Theta$.

In practice, we do not define a constraint by giving the family of functions $\Theta$ extensionally, as the size of such set might be exponential in the number of users and steps. 
Instead, we will assume that constraints have ``compact'' descriptions, in the sense that it takes polynomial time to test whether a given plan satisfies a constraint. 
This is a reasonable assumption, as all constraints of relevance in practice satisfy this property.
For instance, the two most well-known constraints are \emph{binding-of-duty} (BoD) and \emph{separation-of-duty} (SoD). 
The \emph{scope} of these constraints is binary: a plan $\pi$ satisfies a BoD constraint $(\{s_1,s_2\},=)$ if and only if $\pi(s_1) = \pi(s_2)$; and $\pi$ satisfies an SoD constraint $(\{s_1,s_2\},\ne)$ if and only if $\pi(s_1) \neq \pi(s_2)$.
A natural generalization of these constraints are $\atmost$ and $\dsod$ constraints, in which the scope may be of arbitrary size, and the definition of such constraints includes an additional integer $k$. 
Given $T \subseteq S$, a plan satisfies $\atmost(T, k)$ (resp. $\dsod(T, k)$) if and only if $|\pi(T)| \le k$ (resp. $|\pi(T)| \ge k$). 

\emph{User-independent constraints} generalize all these forms of constraints~\cite{CoCrGaGuJo14}.
Informally, such a constraint limits the execution of steps in a workflow, but is indifferent to the particular users that execute the steps.
More formally, a constraint $(T, \Theta)$ is user-independent if whenever $\theta \in \Theta$ and $\psi : U \rightarrow U$ is a permutation then $\psi \circ \theta \in \Theta$ (where $\circ$ denotes function composition). 
A separation of duty constraint, on two steps for example, simply requires that two \emph{different} users execute the steps, not that, say, Alice and Bob (in particular) must execute them.
Similarly, a binding of duty constraint on two steps only requires that the \emph{same} user executes the steps.
More generally, $\dsod$ and $\atmost$ constraints are user-independent. It appears most constraints that are useful in practice are user-independent: all constraints defined in the ANSI-RBAC standard~\cite{ansi-rbac04}, for example, are user-independent.
A simple workflow specification, which will be used as a running example, can be found in Appendix~\ref{section:workflow-example}. 

A \emph{constrained workflow authorization schema} is a tuple $(G=(S, E), U, A, C)$, where $(G, U, A)$ is a workflow schema, and $C$ is a set of constraints.
We say that a plan $\pi : S \rightarrow U$ is \emph{authorized} if $(s, \pi(s)) \in A$ for every $s \in S$, and we say that $\pi$ is \emph{valid} if it is 
 authorized
 and if it satisfies all $c \in C$.
Then the \WSPlong is defined in the following way~\cite{WaLi10}:

\begin{center}
\fbox{%
      \begin{tabulary}{.95\columnwidth}{@{}r<{~}@{}L@{}}
        \multicolumn{2}{@{}l}{\WSPlong (\WSP)}\\
        \emph{Input:} & A constrained workflow authorization schema $W = (G=(S, E), U, A, C)$\\
        \emph{Question:} & Is there a valid plan $\pi : S \rightarrow U$?
       \end{tabulary}%
      }
\end{center}

Henceforth, ``workflow schema'' will mean ``constrained workflow authorization schema''.

\subsection{Valued workflow satisfiability}\label{sec:vwsp}

We now review the problem of minimizing the cost of ``breaking'' the policies and/or constraints.
Informally, given a workflow schema, for each plan $\pi$, we define the weight (total cost)  $w(\pi)$ associated with the plan $\pi$.
The problem, then, is to find a  plan with minimum total cost.

More formally, let $((S,E),U,A,C)$ be a workflow schema.
Let $\Pi$ denote the set of all possible plans from $S$ to $U$.
Then, for each $c \in C$, we define a weight function $w_c : \Pi \rightarrow \mathbb{Z}$, where
 \[
  w_c(\pi)
   \begin{cases}
    = 0 & \text{if $\pi$ satisfies $c$}, \\
    > 0 & \text{otherwise}.
   \end{cases}
 \]
The pair $(c,w_c)$ is a \emph{weighted constraint}.
Then we define the \emph{constraint weight} of $\pi$ to be
\[
 w_C(\pi) = \sum_{c \in C} w_c(\pi).
\]
Note that $w_C(\pi) = 0$ if and only if $\pi$ satisfies all constraints in $C$.

We next introduce a function $w_A : \Pi \rightarrow \mathbb{Z}$, which assigns a cost for each plan with respect to the authorization policy.
The intuition is that a plan in which every user is authorized for the steps to which she is assigned has zero cost and the cost of a plan that violates the policy increases as the number of steps that are assigned to unauthorized users increases.
More formally, we define  the \emph{authorization weight} of $\pi$ to be
\[
 w_A(\pi)
  \begin{cases}
   = 0 & \text{if $(t,\pi(t)) \in A$ for all $t\in S$}, \\
   > 0 & \text{otherwise}.
  \end{cases}
\]
%

Then the valued {\sc Valued Workflow Satisfiability Problem} is defined in the following way~\cite{CrGuKa15b}.

\begin{center}
\fbox{%
      \begin{tabulary}{.95\columnwidth}{@{}r<{~}@{}L@{}}
        \multicolumn{2}{@{}l}{\sc Valued WSP}\\
        \emph{Input:} & A constrained workflow authorization schema $((S,E),U,A,C)$ with weights for constraints and authorizations, as above.\\
        \emph{Output:} & A plan $\pi : S \rightarrow U$ that minimizes $w(\pi)=w_C(\pi) + w_A(\pi)$.
       \end{tabulary}%
      }
\end{center}

Under the assumption that for every plan $\pi$ its weight $w(\pi)$ can be computed in time polynomial in $|S|+|U|+|C|$ (this assumption is justified in many practical situations \cite{CrGuKa15b}), Crampton et al.~\cite{CrGuKa15b} proved the following:

\begin{theorem}\label{thm:VWSP}
{\sc Valued WSP} can be solved in time $\cO^*(2^{k\log k}),$ where $k=|S|.$
\end{theorem}

\subsection{Compositional workflows}\label{sec:def-compositional-workflow}

This section summarizes the model introduced by Crampton, Gutin and Watrigant~\cite{CrGuWa17}.
A \emph{compositional workflow specification} is defined recursively using three operations: serial composition, parallel branching and xor branching. 
Like a ``classical" workflow specification, it can be represented as a DAG $G=(V, E)$. However, in the case of a compositional workflow, not all vertices represent steps. 
In addition to the set of (classical) steps, $V$ also contains $R$, the set of \emph{release points}~\cite{BaBuKa14}, and $O$, the set of \emph{orchestration points}. 
We will sometimes directly define a compositional workflow specification as $G=(S\cup R\cup O, E)$.

The DAG of a compositional workflow always contains two special orchestration points: a source vertex $\alpha$, called \emph{input} and a sink vertex $\omega$, called \emph{output}. 
Moreover, an \emph{atomic} compositional workflow specification (\ie the base case for constructing such a workflow) consists of a single step or release point $v$, and can be represented by the DAG $G=(\{\alpha, v, \omega\}, \{(\alpha, v), (v, \omega)\})$.
Given two compositional workflows $G_1 = (V_1, E_1)$ and $G_2=(V_2, E_2)$ with respective input and output vertices $\alpha_1, \omega_1$ and  $\alpha_2, \omega_2$, respectively, we may construct new compositional workflows using serial composition, and parallel branching and xor branching, denoted by $G_1;G_2$, $G_1 \parallel G_2$ and $G_1 \otimes G_2$, respectively. 
We assume that $V_1 \cap V_2 = \emptyset$.

For \emph{serial composition}, all the steps in $G_1$ must be completed before the steps in $G_2$. 
Hence, the DAG of $G_1;G_2$ is formed by taking the union of $V_1$ and $V_2$, the union of $E_1$ and $E_2$, and the addition of a single edge from $\omega_1$ to $\alpha_2$. Thus, $\alpha_1$ (resp. $\omega_2$) is the input (resp. output) vertex of $G_1;G_2$.

For \emph{parallel composition}, the execution of the steps in $G_1$ and $G_2$ may be interleaved. 
Hence, the DAG of $G_1 \parallel G_2$ is formed by taking the union of $V_1$ and $V_2$, the union of $E_1$ and $E_2$, the addition of new input and output vertices $\alpha_{\parallel}$ and $\omega_{\parallel}$, and the addition of edges $(\alpha_{\parallel}, \alpha_1), (\alpha_{\parallel}, \alpha_2), (\omega_1, \omega_{\parallel})$ and $(\omega_2, \omega_{\parallel})$. 
This form of composition is sometimes known as an \emph{AND-fork}~\cite{AaHoKiBa03} or a \emph{parallel gateway}~\cite{WhMi08}.

In both serial and parallel compositions, all steps in $G_1$ and $G_2$ are executed. 
In \emph{xor composition}, either the steps in $G_1$ are executed or the steps in $G_2$, but not both. 
In other words, xor composition represents non-deterministic choice in a workflow specification. 
The DAG $G_1 \otimes G_2$ is formed by taking the union of $V_1$ and $V_2$, the union of $E_1$ and $E_2$, the addition of new input and output vertices $\alpha_{\otimes}$ and $\omega_{\otimes}$, and the addition of edges $(\alpha_{\otimes}, \alpha_1), (\alpha_{\otimes}, \alpha_2), (\omega_1, \omega_{\otimes})$ and $(\omega_2, \omega_{\otimes})$. 
Given $G_1 \otimes G_2$, we will say that every pair of vertices $(v, v') \in V_1 \times V_2$ are \emph{exclusive}. 
We say that a compositional workflow is \emph{xor-free} if it can be constructed with only serial and parallel operations.

For the sake of readability, we will sometimes simplify the representation of a compositional workflow by replacing an orchestration point having a single in-neighbor $u$ and a single out-neighbor $v$ by the edge $(u, v)$ (for instance, a path $(\alpha_1, s_1, \omega_1, \alpha_2, s_2, \omega_2)$ will be replaced by $(\alpha_1, s_1, s_2, \omega_2)$).

A compositional workflow specification $G=(V, E)$ together with an authorization policy $A \subseteq S \times U$ will be called a \emph{compositional workflow schema}.
An example of a compositional workflow specification is shown in Figure~\ref{fig:example-compositional-workflow}.

\subsubsection{Execution sequences}

In a compositional workflow having an xor branching, there exists more than one set of steps that could comprise a workflow instance. 
And in a compositional workflow having only parallel branching, two different workflow instances will contain the same steps but they may occur in different orders. 
We characterize the set of possible workflow instances in terms of \emph{execution sequences}.

An execution sequence is an ordered sequence of steps and release points and may be empty.
For execution sequences $\sigma = (a_1, \dots, a_k)$ and $\sigma' = (b_1, \dots, b_k)$, we define the following two sets of execution sequences:
 \begin{align*}
  \sigma \sqconcat \sigma' = {} &\{(a_1,\dots,a_k,b_1,\dots,b_\ell)\} \\
  \sigma \sqinter \sigma' = {} &\{(a_1) \sqconcat \sigma'' : \sigma'' \in (a_2,\dots,a_k) \sqinter (b_1,\dots,b_\ell)\} \cup {} \\
						& \{(b_1) \sqconcat \sigma'' : \sigma'' \in (a_1,\dots,a_k) \sqinter (b_2,\dots,b_\ell)\} \\
  \sigma \sqinter () = {} & () \sqinter \sigma = \sigma 						
 \end{align*}
 
In other words, $\sigma \sqconcat \sigma'$ represents concatenation of $\sigma$ and $\sigma'$; and $\sigma \sqinter \sigma'$ represents all possible interleavings of $\sigma$ and $\sigma'$ that preserve the ordering of elements in both $\sigma$ and $\sigma'$.
Given sets of execution sequences $\Sigma$ and $\Sigma'$, we write $\Sigma \sqconcat \Sigma'$ and $\Sigma \sqinter \Sigma'$ to denote $\{\sigma \sqconcat \sigma' : \sigma \in \Sigma, \sigma' \in \Sigma'\}$ and $\{\sigma \sqinter \sigma' : \sigma \in \Sigma, \sigma' \in \Sigma'\}$, respectively.

For a compositional workflow $G$, we write $\Sigma(G)$ to denote the set of execution sequences for $G$.
Then:
\begin{itemize}
\item for workflow specification $G$ comprising a single step or release point $v$, $\Sigma(G) = \{(v)\}$;
\item $\Sigma(G_1 ; G_2) = \Sigma(G_1) \sqconcat \Sigma(G_2)$;
\item $\Sigma(G_1 \parallel G_2) = \Sigma(G_1) \sqinter \Sigma(G_2)$; and
\item $\Sigma(G_1 \otimes G_2) = \Sigma(G_1) \cup \Sigma(G_2)$.
\end{itemize}

For an execution sequence $\sigma$, let $\sigma_S$ and $\sigma_R$ be the restriction of $\sigma$ to the set of steps and release points, respectively. 
Similarly, let 
$S(\sigma)$ and $R(\sigma)$ be respectively the set of steps and release points contained in $\sigma$.%
\footnote{Hence, the difference between $\sigma_S$ and $S(\sigma)$ (resp. $\sigma_R$ and $R(\sigma)$) is that the former is                 an ordered sequence, while the latter is a set. In particular, it might be the case, for two ordered sequences $\sigma$, $\sigma'$, that, say, $S(\sigma) = S(\sigma')$ while $\sigma_S \neq \sigma'_S$, in the case where $\sigma$ and $\sigma'$ are two different orderings of a same set of steps.}

Given an execution sequence $\sigma=(v_1,\dots , v_n)$ of $G$ and $i\in [n]$, we define $\l_{\sigma}(v_i)=(v_1,\dots ,v_{i-1})$, $\r_{\sigma}(v_i)=(v_{i+1},\dots ,v_n)$. Also, if $1\le i<j\le n$, then define 
$\between_{\sigma}(v_i, v_j)=(v_{i+1},\dots ,v_{j-1}).$ We will omit the $\sigma$ subscript from $\l_{\sigma}$, $\r_{\sigma}$ and $\between_{\sigma}$ when it is obvious from context.

\subsubsection{Constraints with release points}\label{sec:constraints}
Suppose we have a requirement that two steps $s_1$ and $s_2$ be performed by the same user if a certain instance-specific condition holds; and they should be performed by different users otherwise.
In other words, the constraint on the execution on $s_1$ and $s_2$ varies depending on the instance.
Release points can be used to encode such requirements by positioning different release points in different, mutually-exclusive branches of the workflow and specifying both constraints on the two steps.
Then passing through one branch ``switches off'' the separation-of-duty constraint, while passing through the other branch switches off the binding-of-duty constraint.

The model for constraints with release points described below~\cite{CrGuWa17} is more general than that of Basin, Burri and Karjoth~\cite{BaBuKa14}.
Let $W = (S\cup R\cup O, E, U, A)$ be a compositional workflow schema. 
A \emph{constraint with release points} has the form $c=(T, \Theta, P)$, where $T \subseteq S$ is the scope of the constraint, $P \subseteq R$ represents the release points of the constraints, and $\Theta$ is a family of functions with domain $T$ and range $U$. 
For $Q \subseteq S$, we denote by $\Theta|_{Q} = \{f|_{Q} : f \in \Theta\}$ the restriction of the family $\Theta$ to $Q$.

Let $\sigma$ be an execution sequence of $W$, and $\sigma_P = (r_1, \dots, r_q)$ be the ordering of release points of $P$ in $\sigma$. 
For every $i \in \{1, \dots, q-1\}$, define
\begin{align*}
 T_0 &= T \cap S(\l(r_1)); \\
 T_i &= T \cap S(\between(r_i, r_{i+1})),\ \text{for $i \in [q-1]$}; \\
 T_q &= T \cap S(\r(r_q)).
\end{align*}
In other words, for $i \in [q-1]$, $T_i$ is the set of steps of $T$ occurring between $r_i$ and $r_{i+1}$ in $\sigma$.

Given a constraint $c = (T,\Theta,P)$ and an execution sequence $\sigma$, we define the \emph{restriction} of $c$ to $T_i$ to be the constraint $c_i = (T_i, \Theta|_{T_i})$.
(That is, a constraint with scope limited to $T_i$ and having no release points.)
We say that a plan $\pi : S(\sigma) \rightarrow U$ \emph{satisfies} $c$ if and only if for all $i \in \{0, \dots, q\}$, $\pi|_{T_i}$ satisfies $c_i$, \ie if $\pi|_{T_i} \in \Theta|_{T_i}$. 
Informally, a plan satisfies $c$ if and only if its restriction to each subscope $T_i$, $i \in \{0, \dots, q\}$, can be extended to a valid tuple (\ie a tuple which belongs to $\Theta$). 
We say $\sigma$ \emph{satisfies} $c$ if there exists a plan $\pi : S(\sigma) \rightarrow U$ that satisfies $c$.
%
%

A constrained compositional workflow schema (\ccws for short) is a tuple  $(G=(S\cup R\cup O, E), U, A, C)$, where $(G, U, A)$ is a compositional workflow schema, and $C$ is a set of constraints with release points. 
We assume the scope of a constraint does not contain two exclusive steps. 
This is a reasonable assumption since two exclusive steps never occur in the same execution sequence.
We say constraint $c=(T, \Theta, P)$ is user-independent (UI) if and only if for every $\theta \in \Theta$ and every permutation $\phi : U \rightarrow U$, we have $\phi \circ \theta  \in \Theta$. 

\subsubsection{WSP with release points}\label{sec:def-problem}

Given a \ccws $W = (S\cup R\cup O, E, U, A, C)$, we say that an execution sequence $\sigma$ is \emph{satisfied} if there exists an authorized plan $\pi : S(\sigma) \rightarrow U$ that satisfies all constraints in $C$. 
(Note that authorization does not depend on the ordering of steps or release points.)
We say that $W$ is \emph{strongly satisfiable} if and only if every execution sequence of $W$ is satisfiable. 
We then define the following decision problem:

\begin{center}
\fbox{%
      \begin{tabulary}{.95\columnwidth}{@{}r<{~}@{}L@{}}
        \multicolumn{2}{@{}l}{\WSRPP}\\
        \emph{Input:} & A constrained compositional workflow schema  $W = (S\cup R\cup O, E, U, A, C)$\\
        \emph{Question:} & Is $W$ strongly satisfiable ?
       \end{tabulary}%
      }
\end{center}

Clearly \WSRPP is a generalization of \WSP (indeed, a \WSRPP with no xor branching and whose all constraints have no release point is equivalent to a \WSP instance), and is thus $NP$-hard and $W[1]$-hard when parameterized by $k=|S|$ \cite{WaLi10}.
Despite the seeming difficulty of the problem (since all execution sequences have to be considered), \WSRPP is FPT parameterized by the total number of steps, release points and xor-branchings~\cite{CrGuWa17}.

\section{Approximate and Bounded Workflow Satisfiability}\label{sec:bounded-wsp}

There has been considerable interest in recent years in making the specification and enforcement of access control policies more flexible.
Naturally, it is essential to continue to enforce effective access control policies.
Equally, it is recognized that there may well be situations where a simple ``allow'' or ``deny'' decision for an access request may not be appropriate.
It may be, for example, that the risks of refusing an unauthorized request are less significant than the benefits of allowing it.
One obvious example occurs in healthcare systems, where the denial of an access request in an emergency situation could lead to loss of life.
Hence, there has been increasing interest in context-aware policies, such as ``break-the-glass'', which allow different responses to the same request in different situations.
Risk-aware access control is another promising line of research that seeks to quantify the risk of allowing a request, where a decision of ``0'' might represent an unequivocal ``deny'' and ``1'' an unequivocal ``allow'', with decisions of intermediate values representing different levels of risk.

Similar considerations arise very naturally when we consider workflows.
In particular, we may specify authorization policies and constraints that mean a workflow specification is unsatisfiable.
Clearly, this is undesirable from a business perspective, since the business objective associated with the workflow can not be achieved.
Valued WSP provides a way of determining the minimum cost, in terms of violating constraints and/or the authorization policy, of a plan for the given workflow specification.

In this paper, we extend Valued WSP to \ccws{s}.
Given a \ccws, we define for every pair $(\sigma,\pi)$, where $\pi : S(\sigma) \rightarrow U$, a cost $w(\sigma,\pi)$.
In practice, this cost will be determined by the sum of the costs of all constraint and authorization policy violation(s) incurred by the plan, as described in Section~\ref{sec:vwsp}.
We assume $w(\sigma,\pi)$ can be computed in time polynomial in the size of the workflow specification.

Let $f : \Sigma \rightarrow (0,1]$ be a frequency distribution such that 
\[
 \sum_{\sigma \in \Sigma} f(\sigma) = 1,
\]
where $f(\sigma)$ denotes the relative frequency of $\sigma$ occurring as the set of steps in a workflow instance.
The simplest case is a uniform distribution
\[
 f(\sigma) = \frac{1}{\card{\Sigma}},\ \text{for all $\sigma \in \Sigma$}. 
\]
In this case, every execution sequence is equally likely to occur as a workflow instance.
\GGrem{Of course, for some workflow specifications, the uniform distribution will not be appropriate and some execution sequences will be much more likely to occur than others.}
\GGadd{In this paper, we assume that $f$ is distributed uniformly.}

Then, given an execution sequence $\sigma$ and a plan $\pi : S(\sigma) \rightarrow U$, we define $\overline{w}(\sigma,\pi) = f(\sigma) \cdot w(\sigma,\pi)$ to be the \emph{relative cost} of the pair $(\sigma,\pi)$.
We may wish to impose an upper bound on the relative cost of every execution sequence, or bound the expected cost of the workflow, or insist that a particular proportion of workflow instances will have a bounded cost.
More formally, we introduce the following definition.

\begin{definition}
Let $B \geqslant 0$ denote a \emph{budget}.
We say a workflow schema has 
\begin{itemize}
 \item \emph{bounded cost} if for every $\sigma \in \Sigma$, there exists a plan $\pi : S(\sigma) \rightarrow U$ such that \[\overline{w}(\sigma,\pi) \leqslant \frac{B}{\card{\Sigma}}; \]
 \item \emph{bounded expected cost} if for every $\sigma_i \in \Sigma$, there exists a plan $\pi_i : S(\sigma_i) \rightarrow U$ such that \[ \sum_{\sigma_i \in \Sigma} \overline{w}(\sigma_i,\pi_i) \leqslant B.\]
\end{itemize}
\end{definition}

In the special case $B = 0$, a workflow with bounded cost or with bounded expected cost is satisfiable.
And in the special case that $f$ is uniform, bounded cost simply means that for every execution sequence $\sigma$, there exists a plan $\pi$ such that $w(\sigma,\pi) \leqslant B$.
The above definitions naturally give rise to two decision problems: 
\begin{description}
  \item[\sc Bounded Cost WSP (BC-WSP).] Given a workflow specification and a budget does the workflow specification have bounded cost? 
  \item[\sc Bounded Expected Cost WSP (BEC-WSP),] Given a workflow specification and a budget does the workflow specification have bounded expected cost? 
\end{description}
A specification with bounded cost means that the relative cost of every execution sequence can be bounded.
In the special case that $f$ is uniform, the cost of every execution sequence can be bounded by $B$.
A specification with bounded expected cost allows some execution sequences to exceed the budget but the cumulative cost is bounded.  
Such a specification allows for some very rare execution sequences whose only plans are relatively expensive, provided all the more commonly occurring plans have plans that are relatively cheap.

We may also define related search problems:
Given a workflow specification, what is the smallest budget $B$ for which the workflow has \begin{inparaenum}[(i)]\item bounded cost \item bounded expected cost?\end{inparaenum}

\begin{definition}
 Let 
 \[ \Sigma_B = \set{\sigma \in \Sigma : \exists\ \pi : S(\sigma) \rightarrow U, w(\sigma,\pi) \leqslant B} \]  
 denote the set of execution sequences for which there exists a plan with cost no greater than $B$.
 We say a workflow specification has \emph{probability $p$ of completing within budget} if 
\[    \sum_{\sigma \in \Sigma_B} f(\sigma) \geqslant p. \]
\end{definition}

In the simple case that $f$ is uniform, the above definition reduces to 
\[ \frac{\card{\Sigma_B}}{\card{\Sigma}} \geqslant p. \]
Then, we may define a third decision problem called {\sc Approximate BC-WSP}: Given a workflow specification, a budget $B$ and a probability $p$, does the workflow have probability $p$ of completing within budget?

\section{Solving {\sc Bounded Cost {\WSP}}}
\label{sec:bounded-cost-wsp1}
We now describe an algorithm to solve {\sc Bounded Cost {\WSP}}.
Notice that our goal is to determine whether, for every execution sequence, there exists a complete plan with bounded cost.
One naive approach would be to enumerate all execution sequences, and solve {\sc Valued {\WSP}} for each of them.
We show, however, that there is a more efficient way to solve the problem.
We will define an equivalence relation similar to the one defined by Crampton et al.~\cite{CrGuWa17} for the execution sequences, and will see that all equivalent execution sequences have the same weight.

\subsection{Execution arrangements and their enumeration}
We restate the equivalence relation $\sim$ defined in ~\cite{CrGuWa17} for the sake of completeness.
We say that execution sequences $\sigma$ and $\sigma'$ are {\em equivalent} if
\begin{description}
	\item[(i)] $\sigma_R = \sigma'_R$,
	\item[(ii)] $S(\sigma) = S(\sigma')$, and
	\item[(iii)] for all $s \in S, R(right_{\sigma}(s)) = R(right_{\sigma'}(s))$.
\end{description}

Informally, two execution sequences are in the same equivalence class when their release points are the same and occur in the same sequences, they have the same set of steps, and for every step, the set of release points occurring to the right are the same.
Essentially this means that the set of steps occurring between two release points are also the same.
We call each equivalence class, an \emph{execution arrangement}.

Based on the above characterization, we now define a {\em compact} representation of execution arrangement similar to that used by Crampton \emph{et al}~\cite{CrGuWa17}.
For an equivalent class containing an execution sequence $\sigma$, an execution arrangement is an ordered sequence $(S_1,r_1,S_2,r_2,\ldots,r_{q-1},S_q\}$ which satisfies the following properties:

\begin{enumerate}
    \item\label{property-1-execution-arrangement}
$\{S_1,\ldots,S_q\}$ is a partition of $S(\sigma)$. Note that we may have $S_i = \emptyset$ for some $i \in [q]$;
    
    \item\label{property-2-execution-arrangement}
    $(r_1,\ldots,r_{q-1})$ is a linear sub-extension of $G$ containing all release points; and
    
    \item\label{property-3-execution-arrangement} 
    for all $(s_1,\ldots,s_q) \in S_1 \times \ldots \times S_q$, $(s_1,r_1,\ldots,s_{q-1},r_{q-1},s_q)$ is a linear sub-extension of $G$.
\end{enumerate}

\GGadd{For an illustration of compact representations of execution arrangements, see Table~\ref{table:execution-arrangement-illustration}.}

The alert reader will notice that we have abused the notation slightly in the last property if $S_i = \emptyset$ for some $i \in [q]$.
If $S_i = \emptyset$ for some $i \in [q]$, we simply omit such steps $s_i$ in the sequence $(s_1,r_1,\ldots,r_{q-1},s_q)$.
Now, we have the following lemma.

\begin{lemma}
\label{lemma:same-cost}
Let $W = ((S \cup R \cup O, E), U, A, C)$ be a CCWS.
Let $\sigma$ and $\sigma'$ be two execution sequences and $\sigma \sim \sigma'$.
Then, for a plan $\pi: S(\sigma) \rightarrow U$, we have that $w_{\sigma}(\pi) = w_{\sigma'}(\pi)$.	
\end{lemma}

\begin{proof}
Let $c = (T, \Theta, R)$ be a constraint.
By definition of $\sim$, we have that $\sigma_R = \sigma'_R = (r_1,\ldots,r_q)$.
Now, let $i \in [q-1]$, and denote by $T_i$ the set $T \cap S(btw_{\sigma}(r_i, r_{i+1}))$, and by $T'$ the set $T \cap S(btw_{\sigma'}(r_i, r_{i+1}))$.
We know by definition of $\sim$ that $R(right_{\sigma}(s)) = R(right_{\sigma'}(s))$ for every $s \in \sigma'$.
Hence, constraint $c$ is violated by a plan $\pi$ in $\sigma$ if and only if the constraint $c$ is violated by a plan $\pi$ in $\sigma'$.
Also, authorization does not depend on the ordering of steps or release points.
So, if a particular authorization policy is violated by $\pi$ in $\sigma$, then the same authorization policy is violated by $\pi$ in $\sigma'$ as well.
Therefore we have that $w_{\sigma}(\pi) = w_{\sigma'}(\pi)$.
\end{proof}

From Lemma~\ref{lemma:same-cost}, the following is clear.
Suppose that we find a plan $\pi$ of  minimum cost for each of the execution arrangements.
Then, this will be sufficient to find a plan for each of the execution sequences.
So, we may proceed as follows. First, we enumerate the set of execution arrangements.
After that, for each of the execution arrangements, we find a plan of minimum cost. If all such plans are of weight bounded by $B$,
then {\sc Bounded Cost {\WSP}} has a solution; otherwise, it has no solution. Below we will describe details of this approach.

It follows from~\cite{CrGuWa17} that the number of execution arrangements is $(|S| + |R|)!$, and their enumeration is non-trivial.
The presence of xor branching means that the set of steps and release points differ depending on the executions.
We can use the approach of~\cite{CrGuWa17} to enumerate all execution arrangements.
We provide an outline of the enumeration, rather than describing it in detail; a full description is available in~\cite{CrGuWa17}.
The overall approach is decomposed into two parts:

\begin{enumerate}
	\item\label{xor-branching-elimination} elimination of xor branching~\cite[Section 3.2]{CrGuWa17}; and
	\item\label{execution-arrangement-elimination} enumeration of all execution arrangements in an xor branching~\cite[Algorithm 2]{CrGuWa17}.
\end{enumerate}

Let $G$ be the DAG of the initial workflow instance $W$.
After eliminating xor-branchings, we obtain a collection of workflow instances $W[G_1],W[G_2],\ldots,W[G_t]$.
Then $W$ is satisfiable if and only if for every $i \in [t]$, $W[G_i]$ is satisfiable~\cite[Lemma 3.2]{CrGuWa17}.
The reason here is that an execution of $W[G_i]$ is also an execution sequence of $W$.
Also, for every execution sequence $\sigma$ of $W$, there exists $i \in [t]$ such that $\sigma$ is an execution sequence of $W[G_i]$.
Hence, consider an arbitrary $i \in [t]$, and consider an arbitrary execution sequence $\sigma$ of $W[G_i]$.
If a plan $\pi$ has weight $w^\star$ for $\sigma$ in $W[G_i]$, then $\pi$ also has weight $w^\star$ for the execution sequence $\sigma$ in $W$ since the execution sequences are same.
Hence,  if a plan $\pi$ for an execution sequence in an xor-free instance has bounded weight, then the same plan also has the weight with same bound in the original workflow instance.
Thus, it is sufficient to compute the execution arrangements for each of the xor-free instances and find if there is a plan with bounded weight for each of the execution arrangements.

The existing algorithm for enumerating execution arrangements can be used, unchanged, as it applies to workflow instances containing release points and orchestration points, and is independent of whether costs are associated with policy and constraint violations. We may also make use of the following theorem~\cite[Theorem 3.5]{CrGuWa17}.

\begin{theorem}
\label{lemma:enumaration-of-execution-arrangement-algorithm}
Given an instance of CCWS $W = ((S \cup R \cup O, E), U, A, C)$ with release points, there exists an algorithm that removes all xor-branchings, and enumerates the set of all execution arrangements in time $\cO^*(2^{|\Bxor|} |R|! (|R|+1)^{|S|})$.
\end{theorem}	

Having enumerated the execution arrangements, we must determine whether there exists a plan $\pi$ with bounded cost for each of the execution arrangements.
We consider this problem in the next subsection.

\subsection{Reduction to Valued WSP}
Now for each execution arrangement, we have to determine whether there exists a plan $\pi$ with bounded cost.
We show that this can be reduced to solving finitely many instances of the classical {\sc Valued {\WSP}} (stated in Section~\ref{sec:vwsp}), which contains no release points or orchestration points.


Note that for every execution arrangement, we want to solve {\sc Bounded Cost {\WSP}} for one execution sequence.
The idea is similar to the ideas used in~\cite{CrGuWa17}; we describe it here for completeness, and also in order to highlight the differences.

Suppose that $\Sigma = (S_1,r_1,S_2,r_2,\ldots,r_{q-1},S_q)$ is an execution arrangement, and $c = (T, \Theta, P)$ is a constraint with release points $P = \{r_{p_1},\ldots,r_{p_{|P|}}\}$.
We assume without loss of generality that the ordering is a linear extension of release points $R(\Sigma)$.
As before, for all $i \in [P-1]$, we define $T_i = T \cap S(btw(r_{p_i}, r_{p_{i+1}}), T_0 = T \cap S(\mathit{left}(r_{p_0})), T_{|P|} = T \cap S(\mathit{right}(r_{p_{|P|}})$, and the ``classical'' constraint $c_i = (T_i, \Theta|_{T_i})$.
We know that $c$ is satisfied by an execution sequence $\sigma$ if and only if there exists a plan $\pi$ such that $\pi_{T_i}$ satisfies $c_i$ for every $i \in \{0,1,\ldots,q\}$.
Now, suppose that for a plan $\pi$, the constraint $c_i$ is violated by $\pi_{T_i}$ for some $i \in \{0,1,\ldots,q\}$.
Then, the constraint $c$ is also violated by $\pi$.
So, for each $i \in \{1,\ldots,|P|\}$, we define the classical {\sc Valued {\WSP}} instance $W_i = (G_i = (S_i, E_i), U, A_i, C_i)$ that defines the partial order restricted to $S_i$, $A_i = A \cap (S_i \times U)$ and $C_i = \{c_i | c \in C\}$.

For our case, we reduce it to finitely many instances of {\sc Valued WSP}.
Suppose that for every $i \in [|P|]$, we find the weight $w_i^\star$ of a minimum weight plan for the {\sc Valued WSP} $W_i$.
If this $w_i^\star$ meets the bound for every $i \in [|P|]$, then we say that {\sc Bounded Cost WSP} is a yes-instance.
Otherwise we say that {\sc Bounded Cost WSP} is a no-instance.
Hence, we have the following lemma.

\begin{lemma}
\label{lemma:reduction-to-valued-csp}
$\Sigma$ is a yes-instance for {\sc Bounded Cost WSP} if and only if for every $i \in [|P|]$, the cost of the the plan outputted for $W_i$ is at most $B$.
\end{lemma}

\subsection{Running time of the algorithm}
In this subsection, we analyze the running time of the algorithm.
The worst case running time of the algorithm is the respective running times of the algorithm for solving the following subproblems.

\begin{description}
    \item\label{step1} (i) enumerating all xor-free sub-instances;
    \item\label{step2} (ii) given an xor-free instance, enumeration of all execution arrangements;
    \item\label{step2} (iii) given an execution arrangement, reduction to {\sc Valued {\WSP}} and solving {\sc Valued WSP} for each of them.
\end{description}

We know from Theorem~\ref{lemma:enumaration-of-execution-arrangement-algorithm} that the product of steps (i) and (ii) takes $\cO(2^{|\Bxor|} |R|! (|R|+1)^{|S|})$ time.


Hence, as a consequence of Lemma~\ref{lemma:reduction-to-valued-csp} and Theorems \ref{lemma:enumaration-of-execution-arrangement-algorithm} and \ref{thm:VWSP}, we have the following:

\begin{theorem}\label{thm:BCWSP}
\label{thm:bounded-cost-wsp} Let $\kappa=|S|+|R|+|\Bxor|.$
{\sc Bounded Cost {\WSP}} parameterized by $\kappa$ is fixed-parameter tractable, and can be solved in $\cO^*(2^{|\Bxor|} |R|! (|R|+1)^{|S|} |S|^{|S|})$ time.
\end{theorem}

By Theorem \ref{thm:VWSP}, computing minimum weight plans takes time $\cO^*(|S|^{|S|})$ and this computation can contribute a significant factor to the running time of the algorithm of Theorem \ref{thm:BCWSP} especially if the number of xor-branchings and release points is smaller than the number of steps.
In a practical implementation, this aspect of the algorithm can be sped up by observing that the weight of a minimum weight plan $\pi$ is the sum of the weights of minimum weight plans for $S_1,\dots ,S_q$ (see the proof of Lemma \ref{lemma:same-cost}). Thus, each time we compute the weight of minimum weight plan for a set $S_i$, we save it in RAM (e.g., in a bucket of a hash table with key $S_i$) unless $S_i$ is already in RAM and the weight 
of its minimum weight plan can be obtained directly from RAM.


\section{Solving {\sc Bounded Expected Cost {\WSP}} and Approximate BC-WSP}\label{sec:bounded-expected-cost}
We now describe how we can reuse the ideas from Section~\ref{sec:bounded-cost-wsp1} in order to solve {\sc Bounded Expected Cost WSP}.
	Recall that our algorithm for {\sc Bounded Cost WSP} consists of three steps.
	The first two steps -- enumerating all xor-free sub-instances and enumerating execution arrangements -- can be reused without modification.
	The third step, however, requires some changes.

Suppose we have $\card{P}$ instances of {\sc Valued {\WSP}}.
Note that the number of execution sequences in different execution arrangements may vary. 
During the run of Algorithm 2 from~\cite{CrGuWa17}, we can count how many execution sequences there are in a specific execution arrangement.
We store a counter for each of the execution arrangements.
Notice that we finally solve {\sc Valued {\WSP}} for $|P|$ instances where $|P|$ is the number of different execution arrangements.

Recall that $f$ is distributed uniformly over execution sequences.
Let $a_i$ be the number of execution sequences in the $i$th execution arrangement.
Suppose that $w_i^\star$ is the cost of {\sc Valued \WSP} of the $i$th instance.
Then $w_i^\star$ is the cost of {\sc Valued \WSP} for all the execution sequences of the $i$th execution arrangement.
Then the total cost for the $i$th execution arrangement is $a_i w_i^\star$.
Then the total cost over all execution sequences is $$W^*=\frac{\sum\limits_{i = 1}^{|P|} a_i w_i^\star}{\sum_{i = 1}^{|P|} a_i}$$ which is same as the sum of the minimum costs over all execution sequences.
If $W^* \leqslant B$, we say that {\sc Bounded Expected Cost \WSP} is a yes-instance.

Finally, we consider {\sc Approximate BC-WSP.
Recall that we wish to determine whether a workflow schema has probability $p$ of completing within budget $B$.}
We first initialize a counter $b$ to zero where we count the number of execution sequences that achieves a specific bound.
If $w_i^\star \leqslant B$ for the $i$th execution arrangement, we increase the counter by $a_i$ as each of the execution sequences in the $i$th execution arrangement has the same cost.
Finally, $b$ is the number of execution sequences where there exists a plan with cost no greater than $B$.
If $b/\sum_{i = 1}^{|P|} a_i \geq p$, we say that the workflow has the probability $p$ of completing within the budget.

\section{Related work}\label{sec:related-work}

Research on workflow satisfiability began with the seminal work of Bertino, Ferrari and Atluri~\cite{BeFeAt99} and Crampton~\cite{Cr05}.
Wang and Li were the first to demonstrate that WSP, subject to specific and limiting restrictions, was fixed-parameter tractable~\cite{WaLi10}.
A substantial body of work now exists on the fixed-parameter tractability of WSP~\cite{CoCrGaGuJo14,CrGaGuJoWa16,CrGuYe13}.
In particular, it is known that WSP is fixed-parameter tractable (parameterized by the number of steps) when all constraints are regular~\cite{CrGuYe13} or user-independent~\cite{CoCrGaGuJo14}.

Basin, Burri and Karjoth introduced the notion of release points~\cite{BaBuKa14} in order to model workflows in which the set of steps that are executed may vary and for which constraints only apply to certain sets of steps.
They modeled workflows using a process algebra and defined the notion of an \emph{enforcement process}, which corresponds to a valid plan in our model of workflow satisfiability.
They showed that the enforcement process existence (EPE) problem, which corresponds to the workflow satisfiability problem, is NP-hard, and developed a polynomial-time heuristic to solve the EPE problem.
Their algorithm achieves good results under the assumption that the user population is large and ``the static authorizations are equally distributed between them''.

Crampton, Gutin and Watrigant~\cite{CrGuWa17} developed FPT algorithms for WSP with release points, provided all constraints were user-independent.
Their algorithms were exact and did not make any assumptions about the distribution of static authorizations.

Crampton, Gutin and Karapetyan~\cite{CrGuKa15b} introduced the {\sc Valued WSP} problem in order to determine  the most suitable plan(s) for a workflow schema, even if that schema was unsatisfiable, based on the cost of constraint and policy violations.
They argued that it may be necessary, perhaps because of business requirements, to allow a workflow to execute, even if it meant that some constraints or authorizations had to be violated.
Mace, Morisset and van Moorsel~\cite{MaMoMo14} introduced the notion of quantitative resiliency in workflows, which seeks to evaluate how likely a workflow is to complete, given assumptions about the availability of users. 
Crampton, Gutin, Karapetyan and Watrigant~\cite{CrGuKaWa17} extended {\sc Valued WSP} to {\sc Bi-objective WSP} and demonstrated how quantitative resiliency problems could be encoded as instances of {\sc Bi-objective WSP}. 

Bertolissi, dos Santos and Ranise \cite{BertolissiSR18} extended bi-objective WSP optimization to multi-objective.
They used a WSP model similar to that of Crampton, Gutin and Watrigant~\cite{CrGuWa17} but without release points. Thus, one can view \cite{BertolissiSR18} as a earlier approach to combine ideas of \cite{CrGuKa15b} and \cite{CrGuKaWa17}. However, Bertolissi, dos Santos and Ranise \cite{BertolissiSR18} considered only simple separation of duty constraints instead of the general class of user-independent constraints used in our paper. Another difference between \cite{BertolissiSR18} and our paper is that while \cite{BertolissiSR18} uses an off-the-shelf Optimization Modulo Theories solver, we design a specialised FPT algorithm.  Note that several papers including  \cite{CrGuKaWa17} and the very recent \cite{KarapetyanPGG19} showed that FPT algorithms are superior to various off-the-shelf decision and optimization solvers
at solving WSP-related problems.

The problems introduced in this paper continue work on quantitative aspects of workflow satisfiability.
We believe these problems have important practical applications.
Business continuity is hugely important to commercial organizations, arguably more so than security considerations, so the ability to reason about the costs of executing workflows, specified in terms of security violations, is likely to be of significant value.
And it is important that these quantitative problems can be solved by FPT algorithms.


\section{Concluding Remarks}
In this paper, we have introduced an extended model for workflows that incorporates costs for policy and constraint violation and complex workflow patterns.
Some workflows schemas in which the set of steps that are executed may vary from one workflow instance to another may be strongly satisfiable, where every possible instance is satisfiable.
However, there may be schemas having execution sequences that are not satisfiable and it may be impractical, undesirable or impossible, given the personnel available, to define a workflow schema that is strongly satisfiable. On the other hand, business considerations may dictate that the workflow is made operational.
In such circumstances, it is clearly desirable to be know that the cost of executing the workflow is bounded, in terms of constraint and policy violations.
It may also be desirable to be able to determine the maximum cost.

The work in this paper provides the foundations for answering questions of this nature.
Moreover, it provides the theoretical basis for algorithms that can be used to solve such questions relatively efficiently, provided the number of steps is relatively small and only user-independent constraints are employed.

Nevertheless, there remains much interesting work to be done. \GGadd{We would like, for example, to consider relaxing the assumption the uniformity condition on $f$ and
investigate the possibility of obtaining algorithms of similar running time to that in Theorem \ref{thm:BCWSP}.}
We hope to implement the theoretical algorithms described in this paper and evaluate their performance on practical instances of the problems we study.
Prior work on implementing FPT algorithms for WSP-like problem (see, e.g., \cite{KarapetyanGG15}) suggest that the theoretical algorithms described in this paper may be useful in practice, especially if optimizations, such as the one described at the end of Section~5 are employed.

%

\paragraph{Acknowledgement}
The research in this paper was supported by the Leverhulme Trust award RPG-2018-161.

%

\appendix\section{Workflow specification example}\label{section:workflow-example}

We present as a running example a simple purchase-order workflow~\cite{Cr05} in Figure~\ref{fig:example-workflow}.
In the first step of this workflow, the purchase order is created and approved (and then dispatched to the supplier). 
The supplier will submit an invoice for the goods ordered, which is processed by the create payment step. When the supplier delivers the goods, a goods received note (GRN) must be signed and countersigned. 
Only then may the payment be approved and sent to the supplier. 
Observe that this workflow specification contains parallel branches, in the sense that the processing of both $s_3$ and $s_4$ must occur before $s_6$, but the relative ordering of $s_3$ and $s_4$ is of no importance. 
We will extend this example to include mutually exclusive branches.

The workflow specification also includes constraints (each having binary scope), mainly in order to reduce the possibility of fraud.
Such constraints may be depicted as an undirected, labeled graph, in which the vertices represent steps and edges denote constraints, as illustrated in Figure~\ref{fig:example-workflow}(b).
One requirement, for example, is that the steps to create and approve a purchase order are executed by different users.
We will extend the example to include constraints having release points.

\begin{figure}[h]\centering
\subfloat[Ordering on steps]{
\begin{tikzpicture}[->,.=stealth',node distance=8mm and 7mm,semithick,auto]
  \node[draw,rectangle]  (t1)                      {$s_1$};
  \node[draw,rectangle]  (t2) [right=of t1]        {$s_2$};
  \node[draw,rectangle]  (t3) [above right=of t2]  {$s_3$};
  \node[draw,rectangle]  (t4) [below right=of t3]  {$s_4$};
  \node[draw,rectangle]  (t5) [above right=of t4]  {$s_5$};
  \node[draw,rectangle]  (t6) [below right=of t5]  {$s_6$};
  \path (t1) edge (t2)
   (t2) edge (t3)
   (t2) edge (t4)
   (t3) edge (t5)
   (t4) edge (t6)
   (t5) edge (t6);
\end{tikzpicture}}
\hfill
\subfloat[Constraints]{
\begin{tikzpicture}[-,node distance=8mm and 8mm,semithick,auto]
  \node (t1) {$s_1$};
  \node (t2) [above=of t1] {$s_2$};
  \node (t3) [left=of t1] {$s_3$};
  \node (t4) [right=of t1] {$s_4$};
  \node (t5) [left=of t3] {$s_5$};
  \node (t6) [right=of t4] {$s_6$};
  \path (t1) edge [dotted] node {$=$} (t3)
        (t3) edge [dotted] node {$\ne$} (t5)
        (t1) edge [dotted] node[swap] {$\ne$} (t4)
        (t1) edge [dotted] node[swap] {$\ne$} (t2)
        (t4) edge [dotted] node[swap] {$\ne$} (t6);
\end{tikzpicture}}
\hfill
\subfloat[Legend]{\footnotesize\setlength{\extrarowheight}{2pt}
  \begin{tabular}{|ll|}
    \hline
    $s_1$ & create purchase order \\
    $s_2$ & approve purchase order \\
    $s_3$ & sign GRN \\
    $s_4$ & create payment \\
    $s_5$ & countersign GRN \\
    $s_6$ & approve payment \\
    \hline
    $\ne$ & different users must perform steps \\
    $=$ & same user must perform steps \\
    \hline
  \end{tabular}}
\caption{A simple constrained workflow for purchase order processing}\label{fig:example-workflow}
\end{figure}
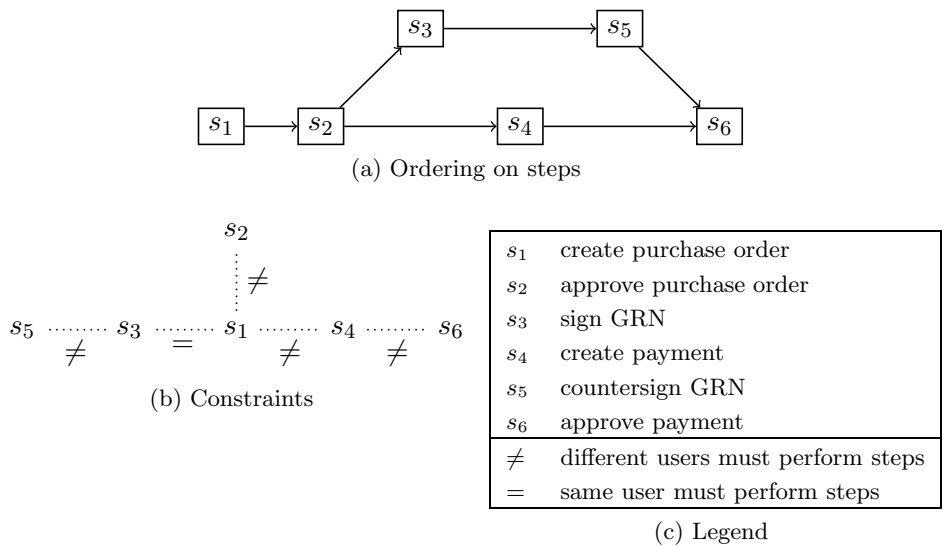

The compositional workflow specification shown in Figure~\ref{fig:example-compositional-workflow} extends the example in Figure~\ref{fig:example-workflow} by including orchestration steps and an xor branching.
We model the fact that orders below a certain value will not require a countersignature on the GRN.
Thus, one branch includes steps to sign and countersign the GRN (which is taken when the value of the order exceeds a certain value), while the other branch contains only the sign GRN step.

\begin{figure}[h]\centering
\subfloat[Ordering on steps]{
\begin{tikzpicture}[->,.=stealth',node distance=5mm and 3mm,semithick,auto]
  \node	(t1)                      	{$\alpha$};
  \node[draw,rectangle]	(t2)	[below=of t1]        	{$s_1$};
  \node[draw,rectangle]	(t3)	[below=of t2]	       	{$s_2$};
  \node	(t4)	[right=of t3]			{$\alpha_{ \parallel }$};
  \node	(t5)	[above right=of t4]		{$\alpha_{\otimes}$};
  \node[draw,rectangle]	(t6)	[above right=of t5]		{$s_3$};
  \node[draw,rectangle]	(t7)[right=of t6]			{$s_5$};
  \node[draw,rectangle]	(t8)	[below right=of t5]		{$s_3'$};
  \node	(t9)	[below right=of t7]		{$\omega_{\otimes}$};
  \node[draw,rectangle]	(t10)[below=of t8]			{$s_4$};
  \node	(t11) [below right=of t9]	{$\omega_{ \parallel }$};
  \node[draw,rectangle] (t12) [right=of t11]			{$s_6$};
  \node (t13) [below=of t12]			{$\omega$};
  \path (t1) edge (t2)
   (t2) edge (t3)
   (t3) edge (t4)
   (t4) 	edge	 (t5)
   (t5) edge (t6)
   (t6) edge (t7)
   (t7) edge (t9)
   (t5) edge (t8)
   (t8) edge (t9)
   (t9) edge (t11)
   (t4) edge (t10)
   (t10) edge (t11)
   (t11) edge (t12)
   (t12) edge (t13);
\end{tikzpicture}}
\hfill
\subfloat[Constraints]{
\begin{tikzpicture}[-,node distance=8mm and 8mm,semithick,auto]
  \node (t1) {$s_1$};
  \node (t2) [above=of t1] {$s_2$};
  \node (t3) [left=of t1] {$s_3$};
  \node (t4) [right=of t1] {$s_4$};
  \node (t5) [left=of t3] {$s_5$};
  \node (t6) [right=of t4] {$s_6$};
  \node (t7) [below=of t1] {$s_3'$};
  \path (t1) edge [dotted] node {$=$} (t3)
        (t3) edge [dotted] node {$\ne$} (t5)
        (t1) edge [dotted] node[swap] {$\ne$} (t4)
        (t1) edge [dotted] node[swap] {$\ne$} (t2)
        (t4) edge [dotted] node[swap] {$\ne$} (t6)
        (t1) edge [dotted] node[swap] {$=$} (t7);
\end{tikzpicture}}
\hfill
\subfloat[Legend]{\footnotesize\setlength{\extrarowheight}{2pt}
  \begin{tabular}{|ll|}
    \hline
    $s_1$ & create purchase order \\
    $s_2$ & approve purchase order \\
    $s_3$ & sign GRN \\
    $s_3'$ & sign GRN \\
    $s_4$ & create payment \\
    $s_5$ & countersign GRN \\
    $s_6$ & approve payment \\
    \hline
    $\ne$ & different users must perform steps \\
    $=$ & same user must perform steps \\
    \hline
  \end{tabular}}
\caption{Example of a compositional workflow specification; vertices with no border represent orchestration points}\label{fig:example-compositional-workflow}
\end{figure}
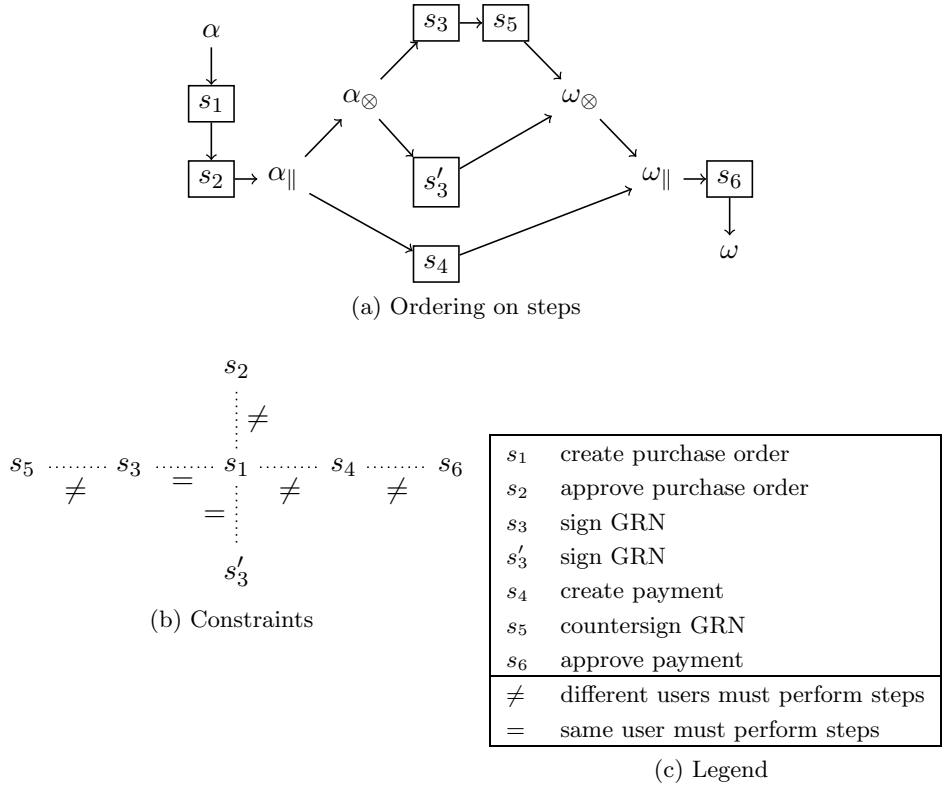

 The possible execution sequences for the example in Figure~\ref{fig:example-compositional-workflow} are:
 \begin{itemize}
	\item $(s_1, s_2, s_4, s_3, s_5, s_6)$
	\item $(s_1, s_2, s_3, s_4, s_5, s_6)$
	\item $(s_1, s_2, s_3, s_5, s_4, s_6)$
	\item $(s_1, s_2, s_4, s'_3, s_6)$
	\item $(s_1, s_2, s'_3, s_4, s_6)$
 \end{itemize}

We extend our running example by modifying the SoD constraint defined between $s_1$ and $s_4$ in order to illustrate how execution sequences and release points might affect the satisfiability of an instance.
The resulting workflow specification is illustrated in Figure~\ref{fig:example-compositional-workflow-releasepoints}.
Specifically, the constraint is released by $r$ positioned between $\omega_{\otimes}$ and $\omega_{\parallel}$.
The intuition is to prevent the same person from creating the purchase order and the payment, except when the GRN has been signed (and countersigned, if the upper branch of the xor branching is chosen).
Hence, if the ``create payment'' is processed before the signature/countersignature of the GRN, then the user who created the purchase order cannot create the payment. 
Otherwise, if the ``create payment'' is processed after the signature/countersignature of the GRN, then the SoD constraint is released.
In the case where the authorization policy is such that only one user is authorized to execute steps $s_1$ and $s_4$, then some execution sequences will be satisfiable, whereas some others will not be satisfiable.

\begin{figure}[h!]\centering
\subfloat[Ordering on steps]{
\begin{tikzpicture}[->,.=stealth',node distance=5mm and 4mm,semithick,auto]
  \node	(t1)                      	{$\alpha$};
  \node[draw,rectangle]	(t2)	[below=of t1]        	{$s_1$};
  \node[draw,rectangle]	(t3)	[below=of t2]	       	{$s_2$};
  \node	(t4)	[right=of t3]			{$\alpha_{ \parallel }$};
  \node	(t5)	[above right=of t4]		{$\alpha_{\otimes}$};
  \node[draw,rectangle]	(t6)	[above right=of t5]		{$s_3$};
  \node[draw,rectangle]	(t7)[right=of t6]			{$s_5$};
  \node[draw,rectangle]	(t8)	[below right=of t5]		{$s_3'$};
  \node	(t9)	[below right=of t7]		{$\omega_{\otimes}$};
  \node[draw,rectangle]	(t10)[below=of t8]			{$s_4$};
  \node	(t11) [below right=of t9]	{$\omega_{ \parallel }$};
  \node[draw,rectangle] (t12) [right=of t11]			{$s_6$};
  \node (t13) [below=of t12]			{$\omega$};
  \node[draw,circle] (t15) [right=of t9]			{$r$};
  \path (t1) edge (t2)
   (t2) edge (t3)
   (t3) edge (t4)
   (t4) 	edge	 (t5)
   (t5) edge (t6)
   (t6) edge (t7)
   (t7) edge (t9)
   (t5) edge (t8)
   (t8) edge (t9)
   (t9) edge (t15)
   (t15) edge (t11)
   (t4) edge (t10)
   (t10) edge (t11)
   (t11) edge (t12)
   (t12) edge (t13);
\end{tikzpicture}}
\hfill
\subfloat[Constraints]{
\begin{tikzpicture}[-,node distance=8mm and 8mm,semithick,auto]
  \node (t1) {$s_1$};
  \node (t2) [above=of t1] {$s_2$};
  \node (t3) [left=of t1] {$s_3$};
  \node (t4) [right=of t1] {$s_4$};
  \node (t5) [left=of t3] {$s_5$};
  \node (t6) [right=of t4] {$s_6$};
  \node (t7) [below=of t1] {$s_3'$};
  \path (t1) edge [dotted] node {$=$} (t3)
        (t3) edge [dotted] node {$\ne$} (t5)
        (t1) edge [dotted] node[swap] {$\ne_r$} (t4)
        (t1) edge [dotted] node[swap] {$\ne$} (t2)
        (t4) edge [dotted] node[swap] {$\ne$} (t6)
        (t1) edge [dotted] node[swap] {$=$} (t7);
\end{tikzpicture}}
\hfill
\subfloat[Legend]{\footnotesize\setlength{\extrarowheight}{2pt}
  \begin{tabular}{|ll|}
    \hline
    $s_1$ & create purchase order \\
    $s_2$ & approve purchase order \\
    $s_3$ & sign GRN \\
    $s_3'$ & sign GRN \\
    $s_4$ & create payment \\
    $s_5$ & countersign GRN \\
    $s_6$ & approve payment \\
    \hline
    $r$ & release point of the constraint blue $(s_1, s_4, \ne)$\\
    \hline
    $\ne$ & different users must perform steps \\
    $=$ & same user must perform steps \\
  	$\ne_r$ & same as $\ne$ but released by $r$\\
    \hline
  \end{tabular}}
\caption{A constrained compositional workflow specification with release points; vertices bordered by a rectangle (resp. circle) represent steps (resp. release points); vertices with no border are orchestration points.}\label{fig:example-compositional-workflow-releasepoints}
\end{figure}

The compact representation of execution arrangements (see Section~\ref{sec:bounded-cost-wsp1}) for the example of Figure~\ref{fig:example-compositional-workflow-releasepoints} are illustrated in Table~\ref{table:execution-arrangement-illustration}.

\begin{table}[h!]
\centering
\begin{tabular}{|c|c|}
\hline
{\bf Arrangement} & {\bf Sequence}\\
\hline
$\{s_1,s_2,s_3,s_5\},r,\{s_4,s_6\}$ & $(s_1,s_2,s_3,s_5,r,s_4,s_6)$\\
\hline
$\{s_1,s_2,s_3,s_4,s_5\},r,\{s_4,s_6\}$ & $(s_1,s_2,s_3,s_4,s_5,r,s_6)$\\
~ & $(s_1,s_2,s_3,s_5,s_4,r,s_6)$\\
~ & $(s_1,s_2,s_4,s_3,s_5,r,s_6)$\\
\hline
$\{s_1,s_2,s_3',s_4\},r,\{s_6\}$ & $(s_1,s_2,s_3',s_4,r,s_6)$\\
~ & $(s_1,s_2,s_4,s_3',r,s_6)$\\
\hline
$\{s_1,s_2,s_3'\},r,\{s_4,s_6\}$ & $(s_1,s_2,s_3',r,s_4,s_6)$\\
\hline
\end{tabular}
\caption{Illustration of Execution Arrangements}
\label{table:execution-arrangement-illustration}
\end{table}

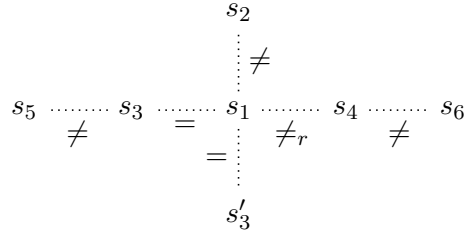
\end{document}